\documentclass[12pt,a4paper]{article}

\usepackage{latexsym}
\usepackage{amsmath}
\usepackage{amssymb}
\usepackage{amscd}
\usepackage{graphicx}
\usepackage[bookmarks=false,colorlinks=true,linkcolor=black,citecolor=black,filecolor=black,urlcolor=black]{hyperref}

\newtheorem{theo}{Theorem}[section]

\newtheorem{prop}[theo]{Proposition}
\newtheorem{defi}[theo]{Definition}

\newenvironment{proof}{\indent{\em Proof.~}}
{\hfill{$\Box$\bigskip\newline}}
\newcommand{\EQ}{\begin{equation}}
\newcommand{\EN}{\end{equation}}

\newcommand{\wt}{\operatorname{wt}}
\newcommand{\ba}{{\bf a}}
\newcommand{\bb}{{\bf b}}
\newcommand{\bc}{{\bf c}}

\newcommand{\bh}{{\bf h}}
\newcommand{\bv}{{\bf v}}
\newcommand{\br}{{\bf r}}
\newcommand{\bx}{{\bf x}}
\newcommand{\by}{{\bf y}}
\newcommand{\bz}{{\bf z}}
\newcommand{\bo}{{\bf 0}}

\newcommand{\F}{\mathbb{F}}

\title{On binary quadratic symmetric bent and almost bent functions\thanks{This work has been partially
        supported by the Spanish grant TIN2016-77918-P (AEI/FEDER, UE).
        The research of the second author was carried out at the IITP RAS at the expense of the Russian Fundamental Research Foundation (project No. 19-01-00364)).
        \newline \indent $^1$Department of
        Information and Communications Engineering, Universitat
        Aut\`{o}noma de Barcelona, 08193-Bellaterra, Spain.
        (email:~josep.rifa@uab.cat) \newline \indent $^2$A. A. Harkevich Institute for Problems of Information
        Transmission, Russian Academy of Sciences, Bol'shoi Karetnyi
        per. 19, GSP-4, Moscow, 101447, Russia. (email: zinov@iitp.ru)}}

\author{J. Rif\`{a}$^1$, V. A. Zinoviev$^2$}

\begin{document}

\maketitle

\begin{abstract}
We give a new simple construction for known binary quadratic symmetric
bent and almost bent functions. In particular, for even number of variables,
they are self-dual and
anti-self-dual quadratic bent functions, respectively, which are not
of the Maiorana-McFarland type, but affine equivalent to it.
\end{abstract}



\section{Introduction}\label{sec1}

Let $\F_q$ be the Galois field of order $q = 2^m$. We use the
standard notation $[n,k,d]$ for a binary linear code $C$ of length
$n$, dimension $k$ and minimum (Hamming) distance $d$. Denote by
$\wt(\bx)$ the Hamming weight of a vector $\bx$ from $\F_2^n$.
For $\bx, \by \in \F_2^n$ denote by $\bx \cdot \by$ the usual inner
product over $\F_2$. Denote by $\bar{\bx}$ the complementary
vector of $\bx$ (obtained by swapping $0$ and $1$).

A {\em boolean function} $f$ in $m$ variables is any map from
$\F_2^m$ to $\F_2$. The {\em weight} of a boolean function $f$
denoted by $\wt(f)$ is the Hamming weight of the binary vector of
the values of $f$, i.e., the number of $\bx \in \F_2^m$ such that
$f(\bx) = 1$. For any boolean function $f$ we define its
Walsh-Hadamard transform $F$, such that \EQ\label{eq:1.1}
F(\ba)~=~\sum_{\bx \in
\F_2^m}\left(-1\right)^{f(\bx)+\ba\cdot\bx},\;\;\forall \,\ba \in
\F_2^m. \EN

\begin{defi}
For even $m$, a boolean function $f$ over $\F_2^m$ is
\textit{bent} if its Walsh-Hadamard transform is $F(\by) =
\pm 2^{m/2}$, for all $\by\in \F_2^m$. For odd $m$, a boolean function
$f$ over $\F_2^m$ is
\textit{almost bent} if its Walsh-Hadamard transform
$F(\by) \in \{0, \pm 2^{(m-1)/2}\}$, for all $\by\in \F_2^m$.
\end{defi}

The {\em nonlinearity} $\mathcal{N}(f)$ of a boolean function $f$ is
its minimum Hamming distance to the set of affine functions, i.e. the value
$\min \{d(f(\bx), \ba \bx+\epsilon)\}$ for all $\ba\in \F_2^m$ and $\epsilon\in \F$.

Two boolean functions $f,g$ are \textit{affine equivalent} if there
exists a linear map $A\in GL(m,\F_2)$, $\bb,\bc \in \F_2^m$ and
$\epsilon\in \F_2$ such that$g(\bx)=f(A(\bx)+\bb)+\bc{\cdot }\bx+\epsilon$,
where $\bc{\cdot }\bx$ is the inner product of $\bc$ and $\bx$.

It is well known that the Walsh-Hadamard transform $F$ of a bent function
$f$ defines a new bent function $\tilde{f}$, such that
$F(\by) = 2^{m/2} (-1)^{\tilde{f}(\by)}$. The function $\tilde{f}$ is
called the \textit{dual} of $f$ and it is fulfilled that $\tilde{\tilde{f}} = f$.
We take the following definitions from [3].

\begin{defi}{[6]}\label{defi:2.1}
A bent function $f$ is called \textit{self-dual}, if it is equal
to its dual. It is called \textit{anti-self-dual}, if it is equal
to the complement of its dual.
\end{defi}

In this quoted paper [1], all self-dual bent functions in $m \le 6$
variables and all quadratic such functions in $m \le 8$ variables
are characterized, up to a restricted form of affine equivalence.
Later, Hou [2] classified all self-dual and anti-self-dual quadratic
bent functions under the action of the orthogonal group $O(m,\F_2)$.

A general class of bent functions is the \emph{Maiorana-McFarland}
class (see [1] and references there), that is, functions of the form
\EQ\label{eq:1.2}
\bx \cdot \varphi(\by) + g(\by),
\EN
where $\bx,\by$ are binary vectors of length $m/2$, $\varphi$ is any
permutation in $\F_2^{m/2}$, and $g$ is an arbitrary boolean function.

One of the interesting class of bent functions are symmetric
functions. An $m$-variable boolean function $f(x)$,~$x=(x_1,\ldots, x_m)$
is called symmetric if $f(x)$ has the same value for any
permutations of variables $x_1,\ldots, x_m$.
Savicky [3] proved that, if $f(x)$ is a symmetric
function of an even number $m$ of variables then the following
statements are equivalent:
\begin{itemize}
\item The function $f$ is a bent function.
\item There exist constants $c,d \in \{0,1\}$ such that for all
$x=(x_1, \ldots, x_m)$
\[
f(x) = \sum_{i<j} x_ix_j + \sum_{i}cx_i + d,
\]
\end{itemize}

This class of quadratic symmetric bent functions was considered in many papers
(see [4 - 7] and references there). It is worth to mention here one result
due to Carlet (see Th. 6 in [4]): \textit{the symmetric functions on $\F_2^m$
which achieve nonlinearity $2^{m-1}-2^{(m-2)/2}$ (i.e. bent) if $m$ is even and
$2^{m-1}-2^{(m-1)/2}$ (i.e. almost bent) if $m$ is odd, are the four
quadratic functions}:
\[
\begin{array}{cccccc}
f_1(x) &=& \binom{\wt(x)}{2} & & &\pmod{2}\\
f_2(x) &=& \binom{\wt(x)}{2} &+ &1& \pmod{2}\\
f_3(x) &=& \binom{\wt(x)}{2} &+ &\wt(x)& \pmod{2}\\
f_4(x) &=& \binom{\wt(x)}{2} &+ &\wt(x)&+ 1 \pmod{2}.\\
\end{array}
\]

One of the open questions concerning self-dual bent and anti-self-dual
bent functions is the following one, quoted in [1]: \emph{are there
quadratic self-dual (anti-self-dual) bent functions which are not of
Maiorana-McFarland type?}

The main result of this note is to give a new description of
symmetric quadratic functions which achieve maximal nonlinearity. This description is simpler than the one previously given by Carlet [4].
Also we give infinite families of self-dual
and anti-self-dual quadratic bent functions which are not of the
Maiorana-McFarland type (but, affine equivalent to it).
The material is organized as follows.
In Section~\ref{sec2} we give some notations and
preliminary results. Section~\ref{sec3} contains the new constructions of
symmetric quadratic  functions and, finally, the
quadratic bent functions which are not of the Maiorana-McFarland
type are given.

\section{Preliminary results}\label{sec2}

Let $m\ge 2$ be an integer and $j \in\{0,1,2,3\}$. Denote by $S(j)_m$
the following sum:
\EQ\label{eq:2.2} S(j)_m~=~\sum_{k =
0,\ldots,m:~k \equiv j \pmod{4}}\binom{m}{k}.
\EN
Denote by
$S(i_1,i_2)_m$ the value $S(i_1,i_2)_m = S(i_1)_m + S(i_2)_m$, for
any two different $i_1$ and $i_2$  from $\{0,1,2,3\}$. The next
proposition was used in [8]
and gives all the values $S(j)_m$ and, hence, the values
of all the sums  $S(i_1,i_2)_m$, which we will use later.
We remark that there are other interesting approaches to reach
the same result like, for instance, in [9] using
Krawtchouk polynomials.

\begin{prop}{[8]}\label{calculS}
For any $m\ge 2$ denote $B_m=2^{s-1}$, where $s=\lfloor
\frac{m}{2}\rfloor$.  Then, the values of $S(j)_m$  depending on
$m \equiv 0,1,2,3,4,5,6,7 \pmod{8}$ are, respectively:

\resizebox{\columnwidth}{!}{
\hspace*{-1cm}\begin{tabular}{ccc|c|c|c|c|c|c|c}
$S(0)_m $&$=$&$ B_m^2+B_m $&$2B_m^2+B_m $&$B_m^2     $&$2B_m^2-B_m $&$B_m^2-B_m $&$2B_m^2-B_m $&$B_m^2     $&$2B_m^2+B_m$\\
$S(1)_m $&$=$&$ B_m^2     $&$2B_m^2+B_m $&$B_m^2+B_m $&$2B_m^2+B_m $&$B_m^2     $&$2B_m^2-B_m $&$B_m^2-B_m $&$2B_m^2-B_m$\\
$S(2)_m $&$=$& $B_m^2-B_m $&$2B_m^2-B_m $&$B_m^2     $&$2B_m^2+B_m $&$B_m^2+B_m $&$2B_m^2+B_m $&$B_m^2     $&$2B_m^2-B_m$\\
$S(3)_m $&$=$&$ B_m^2     $&$2B_m^2-B_m $&$B_m^2-B_m $&$2B_m^2-B_m $&$B_m^2     $&$2B_m^2+B_m $&$B_m^2+B_m $&$2B_m^2+B_m$
\end{tabular}}
\end{prop}

Denote by $H_m$ the binary matrix of size $m\times 2^m$, where the
columns are all different binary vectors of length $m$ and
denote by $\mathcal{H}_m$ the Hadamard $[2^m,m,2^{m-1}]$-code,
generated by $H_m$. Assume that, for any $m\geq 3$, $H_m$ is
constructed using the following recursive way: \EQ\label{eq:2.3}
H_{m+1}=\left[
\begin{array}{c|c}
H_m\,&\,H_m\\
0\ldots 0\,&\,1\ldots 1
\end{array}
\right]. \EN

For a given $m \geq 3$ and any  $i_1, i_2 \in \{0,1,2,3\}$, where
$i_1 \neq i_2$, denote by $\bv_{i_1,i_2}(m)$ or simply $\bv_{i_1,i_2}$
the binary vector $\bv_{i_1,i_2}=(v_0,v_1, ..., v_{n-1})$ whose $j$-th position
$v_j$ is a function of the value of weight of the column $\bh_j$
in $H_m$:
\[
v_j~=~\left\{
\begin{array}{ccc}
1, ~~&\mbox{if}&~~\wt(\bh_j) \equiv i_1~\mbox{or}~ i_2 \pmod{4}\\
0, ~~&  &  ~~\mbox{otherwise}.
\end{array}
\right.
\]

\begin{prop}\label{weights}
The weight distribution of the coset $\mathcal{H}_m+\bv_{i_1,i_2}$ is
$$
\left\{
\begin{array}{lll}
\{2^{m-1}\pm~2^{\frac{m-2}{2}}\} & \mbox{ if and only if} & \mbox{$m$ is even and } i_1-i_2 \equiv 1 \pmod{2},\\
\{2^{m-1}, 2^{m-1}\pm 2^{\frac{m-1}{2}}\} & \mbox{ if and only if} & \mbox{$m$ is odd and } i_1-i_2 \equiv 1 \pmod{2}.
\end{array}
\right.
$$
\end{prop}
\begin{proof}
For even $m$ it follows  from Proposition~\ref{calculS} and, in
particular, can also be found in [8].

We explain the case of odd $m$. Let $m=2u+1$.  By definition
of $\bv_{i_1,i_2}(m)$ we have \EQ\label{eq:2.4}
\bv_{i_1,i_2}(2u+1) =
(\bv_{i_1,i_2}(2u)\,|\,\bv_{i_1-1,i_2-1}(2u)), \EN where
subtraction of indices is modulo $4$. Denote by $\br^{(i)}$ the
$i$-th row of $H_m$, where $i=1,\ldots, m$ and let
\[
\br = \sum_{s=1}^m c_i\br^{(i)},\;\;c_i\in \F_2
\]
be an arbitrary codeword of $\mathcal{H}_m$.
Now from (\ref{eq:2.4}), we deduce
\[
d(\bv_{i_1,i_2}(2u+1), \br)= \left\{
\begin{array}{ccc}
2 d(\bv_{i_1,i_2}(2u),
\br),\;&\mbox{if}&\;c_m=0,\\
2^{2u},\;&\mbox{if}&\;c_m=1.
\end{array}
\right.
\]
For $m=2u$ we have [8] for the case $i_1-i_2 \equiv 1 \pmod{2}$
\[
d(\bv_{i_1,i_2}(2u), \br) \in \{2^{2u-1} \pm 2^{(2u-2)/2}\}.
\]
Since $2(2^{2u-1} \pm 2^{u-1}) = 2^{2u} \pm 2^u = 2^{m-1}\pm 2^{\frac{m-1}{2}}$,
we can deduce that only for the case $i_1-i_2 \equiv 1
\pmod{2}$, the weight distribution of the coset
$\mathcal{H}_m+\bv_{i_1,i_2}$ is $\{2^{m-1}, 2^{m-1}\pm 2^{\frac{m-1}{2}}\}$.
Indeed, the case $i_1-i_2 \equiv 2 \pmod 2$ gives, as possible values
of $d(\bv_{i_1,i_2}(2u), \br)$ [8], either $\{0,2^{m-1}\}$, when
$\{i_1,i_2\}=\{1,3\}$ or $\{2^{m-1},2^m\}$,
when $\{i_1,i_2\}=\{0,2\}$.
\end{proof}

\section{A new construction of quadratic symmetric bent\\
and almost bent functions}\label{sec3}

For any  $i_1, i_2 \in \{0,1,2,3\}$, where $i_1 \neq i_2$, define
the boolean function  $f_{i_1,i_2}$ over $\F_2^{m}$ as:
$$
f_{i_1,i_2}(\bx)=~\left\{
\begin{array}{ccc}
1, ~~&\mbox{if}&~~\wt(\bx) \equiv i_1~\mbox{or}~ i_2 \pmod{4}\\
0, ~~&  &  ~~\mbox{otherwise}.
\end{array}
\right.
$$

Now we state one of the main results of the present paper.

\begin{theo}\label{theo:3.1}
The function $f_{i_1,i_2}$ is quadratic and
$$
\left\{\begin{array}{lll}
\mbox{bent} & \mbox{if and only if} &\mbox{$m$ is even}; m \geq 4 \mbox{ and } i_1 - i_2 \equiv 1 \pmod{2},\\
\mbox{almost bent} & \mbox{if and only if} &\mbox{$m$ is odd}; m \geq 3 \mbox{ and } i_1 - i_2 \equiv 1 \pmod{2}.
\end{array}
\right.
$$
\end{theo}

\begin{proof}
From Proposition~\ref{weights} the function  $f_{i_1,i_2}$ is bent if and only if $i_1 - i_2 \equiv 1 \pmod{2}$. Now, we show that all these bent functions are quadratic. Specifically,
\begin{align*}
f_{2,3}(x_1,x_2,\ldots,x_m) &\equiv\sum_{1\leq i<j\leq m} x_ix_j \pmod{2}\\
f_{1,2}(x_1,x_2,\ldots,x_m) &\equiv \sum_{1\leq i\leq m} x_i + \sum_{1\leq i<j\leq m} x_ix_j\pmod{2}\\
f_{0,1}(x_1,x_2,\ldots,x_m) &\equiv 1+f_{2,3}(x_1,x_2,\ldots,x_m)\pmod{2}\\
f_{0,3}(x_1,x_2,\ldots,x_m) &\equiv 1+f_{1,2}(x_1,x_2,\ldots,x_m)\pmod{2}\\
\end{align*}
To prove the first equality is equivalent to prove that $\binom{w}{2} \equiv 1 \pmod{2}$ if and only if $w\equiv \{2,3\} \pmod{4}$, where $w$ is an integer number, $0\leq w \leq m$. Hence, it is enough to prove that for $w\in\{0,1,2,3\}$, we have $\binom{w}{2} \equiv 1 \pmod{2}$ if and only if $w\in \{2,3\}$, which is clear.

The second equality is reduced to prove that  for $w\in\{0,1,2,3\}$, we have $w+\binom{w}{2} \equiv 1 \pmod{2}$ if and only if $w\in \{1,2\}$, which is also clear.

The last two equalities come from the two first.

Following [5, Ch. 15], we can write a quadratic boolean function as
$f(\bx)= \bx Q \bx^T + L\bx +\epsilon$ and, for the above functions we have:
\begin{equation}\label{theo:2.2}
\begin{split}
f_{2,3}(\bx) &= \bx Q \bx^T, \\
f_{1,2}(\bx) &= \bx Q \bx^T + L\bx^T,\\
f_{0,1}(\bx) &= \bx Q \bx^T + \epsilon,\\
f_{0,3}(\bx) &= \bx Q \bx^T + L\bx^T +\epsilon,
\end{split}
\end{equation}
where $Q$ is the all ones upper triangular binary $m\times m$ matrix with zeroes in
the diagonal, $L$ is the all ones binary vector of length $m$ and $\epsilon = 1$.
\end{proof}

\begin{prop}\label{prop:3.2}
Bent functions $f_{i_1,i_2}$, where $i_1 - i_2 \equiv 1 \pmod{2}$ are affine
equivalent to each other. Bent function $f_{0,1}$ is complementary to $f_{2,3}$,
as well as $f_{0,3}$ is complementary to $f_{1,2}$.
\end{prop}
\begin{proof}
Straightforward from Equation~(\ref{theo:2.2}).
\end{proof}

Notice that the function $f_{0,1}$ and its complementary $f_{2,3}$ has
been constructed in [10] in terms of abelian difference sets, known also
as Menon difference sets (see [11]).

\begin{theo}~

For $m \equiv 0, 4 \pmod{8}$, the functions $f_{i_1,i_2}$ are neither self-dual functions nor
anti-self-dual.

For $m\equiv 2 \pmod{8}$, $f_{2,3}$ and $f_{0,1}$ are self-dual. The function $f_{0,3}$ is
anti-self-dual (with $f_{1,2}$).

For $m\equiv 6 \pmod{8}$, $f_{1,2}$ and $f_{0,3}$ are self-dual. The function $f_{2,3}$ is
anti-self-dual (with $f_{0,1}$).
\end{theo}

\begin{proof}
It is known [1, Th. 4.1] that if $f$ is  a self-dual bent or anti-self-dual bent
quadratic boolean function then the symplectic matrix $Q+Q^T$ associated to $f$
(\ref{theo:2.2}) is an involution, hence $(Q+Q^T)^2 = I$ ($I$ is the identity
matrix of order $m$). Later, Hou [2, Th. 2.1], extended this property to a
necessary and sufficient condition in the sense that $f$ is  self-dual or
anti-self-dual if and only if $(Q + Q^T)^2= I$, and $(Q + Q^T)Q(Q+Q^T) + Q^T$ is
an alternating matrix (so, a matrix of the form $A+A^T$, where $A$ is a square
matrix over $\F_2$).

In all cases of our functions $f_{i_1,i_2}$ the symplectic matrix $Q+Q^T$ coincides
with $J+I$, where $J$ is the $m\times m$ matrix with ones in all the entries and
$Q$ is the all ones upper triangular binary $m\times m$ matrix with zeroes in the
diagonal~(\ref{theo:2.2}).

Hence, $(Q+Q^T)^2=(J+I)^2= J^2 + I = I$ (the last equality is true since $m$ is even).

For the second condition we have
\begin{eqnarray*}
&&(Q + Q^T)Q(Q+Q^T) + Q^T \\
&=& (J+I)Q(J+I) + Q^T \\
&=& JQJ +JQ + QJ +Q + Q^T \\
&=&\binom{m}{2}J +JQ + QJ +Q + Q^T.
\end{eqnarray*}
Taking into account that $m$ is even we see that $\binom{m}{2}J$ is the zero matrix
if and only if $m\equiv 0\pmod{4}$. In all other cases  $\binom{m}{2}J =J$.
We also have $JQ + QJ = (a_{ij})$, where $a_{ij} = 1$ when $i+j$ is even and
$a_{ij}=0$ when $i+j$ is odd. Therefore,  $JQ + QJ = A+A^T+I$, for some $A$.

Finally,
\begin{eqnarray*}
&&(Q + Q^T)Q(Q+Q^T) + Q^T \\
&=&\binom{m}{2}J+A+A^T+I +Q+Q^T\\
&=&\binom{m}{2}J+(A+Q)+(A+Q)^T+I,
\end{eqnarray*}
which is an alternating matrix if and only if $m\not\equiv 0\pmod{4}$ and so, we
conclude that $f_{i_1,i_2}$ is a self-dual or anti-self-dual quadratic bent
function if and only if $m\not\equiv 0\pmod{4}$.

Now that we know in what cases $f_{i_1,i_2}$ is a self-dual or anti-self-dual
quadratic bent function we can check the self-duality or anti-self-duality
condition, for each pair $(i_1,i_2)$. It is not necessary to check that the
dual bent function $\tilde{f}_{i_1,i_2}$ coincides with $f_{i_1,i_2}$ or with
the complement $\bar{f}_{i_1,i_2}$, it is enough to check if the first
coordinate in $f_{i_1,i_2}$ coincides with the first coordinate in $\tilde{f}_{i_1,i_2}$
to decide about self-duality or anti-self-duality.

Let us begin by computing
\[F_{i_1,i_2}(\bo)= \sum_{\bx\in \F_2^m} (-1)^{f_{i_1,i_2}(\bx)+\bo\cdot\bx} =
\sum_{\bx\in \F_2^m} (-1)^{f_{i_1,i_2}(\bx)}= 2^m -2\wt(f_{i_1,i_2})
\]
and note that $\wt(f_{i_1,i_2})$ depends on the value of the pair $(i_1,i_2)$.

Using Proposition~\ref{calculS} we have the following values for $\wt(f_{i_1,i_2}) = S(i_1)_m+S(i_2)_m$, where $B_m = 2^{(m-2)/2}$ :
$$
\begin{array}{|c|c|c|}\hline
\wt(f_{i_1,i_2}) & m\equiv 2  \pmod{8} & m\equiv 6  \pmod{8} \\\hline
(i_1,i_2) = (0,1) & 2B^2_m+B_m& 2B^2_m-B_m\\\hline
(i_1,i_2) = (2,3) &2B^2_m-B_m &2B^2_m+B_m \\\hline
(i_1,i_2) = (0,3) & 2B^2_m-B_m& 2B^2_m+B_m\\\hline
(i_1,i_2) = (1,2) & 2B^2_m+B_m& 2B^2_m-B_m \\\hline
\end{array}
$$

With the above results we can compute $F_{i_1,i_2}(\bo)$, which is always $\pm 2^{m/2}$,
and also $\tilde{f}_{i_1,i_2}(\bo)$, which is in $\{0,1\}$ depending on the value
of $F_{i_1,i_2}(\bo)$.

Now we put in the following tables the values of $\tilde{f}_{i_1,i_2}(\bo)$ and
$f_{i_1,i_2}(\bo)$. When these values coincide we conclude that $f_{i_1,i_2}$ is
self-dual, otherwise $f_{i_1,i_2}$ is anti-self-dual.
\bigskip

\resizebox{\columnwidth}{!}{
$
\begin{array}{|c|c|c||}\hline
 \tilde{f}_{i_1,i_2}(\bo) & m\equiv 2  \pmod{8} & m\equiv 6  \pmod{8} \\\hline
(i_1,i_2) = (0,1) & 1& 0\\\hline
(i_1,i_2) = (2,3) &0&1 \\\hline
(i_1,i_2) = (0,3) & 0&1\\\hline
(i_1,i_2) = (1,2) & 1&0 \\\hline
\end{array}
\begin{array}{||c|c|c|}\hline
f_{i_1,i_2}(\bo) & m\equiv 2  \pmod{8} & m\equiv 6  \pmod{8} \\\hline
(i_1,i_2) = (0,1) & 1& 1\\\hline
(i_1,i_2) = (2,3) &0&0 \\\hline
(i_1,i_2) = (0,3) & 1&1\\\hline
(i_1,i_2) = (1,2) & 0&0 \\\hline
\end{array}
$
}
\end{proof}
\begin{theo}\label{theo:3.3}
For any even $m \geq 4$ the bent functions $f_{i_1,i_2}$, where
$i_1 - i_2 \equiv 1 \pmod{2}$ are not of the
Maiorana-McFarland type.
\end{theo}

\begin{proof}
Let us prove the statement by using contradiction. Assume that
$f_{i_1,i_2}$ is of Maiorana-McFarland type. This means that a binary
variable vector $\bz$, of length $m$, can be divided into
two subvectors $\bx$ and $\by$ of the same length $m/2$ such that
\[
f_{i_1,i_2}(\bz)=\bx\cdot\varphi(\by) + g(\by),
\]
where $\varphi$ is a permutation of $\F_2^{m/2}$ and $g(\by)$ is some
boolean function. Note that $\bx$ and $\by$ run over all the $2^{m/2} \cdot 2^{m/2}$ possible
values. Consider the set of values of $f_{i_1,i_2}(\bz)$ when $\bx$ runs over all
the values in  $\F_2^{m/2}$ and $\by$ is fixed to be $\by=\by_0$, such that $\varphi(\by_0)$
is the zero vector. In this case $x\cdot\varphi(\by_0)=0$ and
$f_{i_1,i_2}(\bz)=g(\by_0)$ is a constant. Since $\bx$ is running over $\F_2^{m/2}$ its
weight takes $m/2+1>2$ different consecutive values. Hence, it takes more
than $2$ different values of weight modulo $4$ and, for all these values,
the function $f_{i_1,i_2}(\bz)$ is constant. This contradicts the definition
of the function $f_{i_1,i_2}$.
\end{proof}

\textbf{Remark 1. } Altough the functions in Theorem~\ref{theo:3.3} are not of the
Maiorana-McFarland type, they are affine equivalent to it [12, Ch. 14,]. Note that
this fact concerning the function $f_{2,3}$  was mentioned in [7].

\end{document}